\RequirePackage{fix-cm}
\documentclass[smallextended]{svjour3}       % onecolumn (second format)
\smartqed  % flush right qed marks, e.g. at end of proof
\usepackage{graphicx}
\usepackage{refcheck} % get out later!
\usepackage{amsmath}
\usepackage{amssymb}
\usepackage{amsfonts}
\usepackage{latexsym}
\usepackage{url}
\usepackage{color}
\usepackage{graphicx}

\newcommand{\E}{\operatorname{E}}
\newcommand{\F}{\operatorname{F}}
\newcommand{\G}{\operatorname{G}}
\newcommand{\Ha}{\operatorname{H}}
\newcommand{\Var}{\operatorname{Var}}
\newcommand{\Cov}{\operatorname{Cov}}
\newcommand{\Prob}{\operatorname{P}}

\newcounter{exa}

\newtheorem{cor}[theorem]{Corollary}

\newtheorem{lem}[theorem]{Lemma}
\newtheorem{ex}[exa]{Example}
\newtheorem{theo}[theorem]{Theorem}

\newtheorem{res}[theorem]{Result}
\begin{document}
\allowdisplaybreaks

\title{Asymptotic normality of the likelihood moment estimators for a stationary linear process with heavy-tailed innovations}


\author{Lukas Martig       \and
        J\"urg H\"usler %etc.
}


\institute{L. Martig \at University of Bern, Inst.\ math.\ Statistics\\
              Sidlerstrasse 5, 3012 Bern, Switzerland \\
              Tel.: +41-31-6318811\\
              Fax: +41-31-6313870\\
              \email{lukasmartig@gmail.com}           %  \\
           \and
           J. H\"usler \at Inst.\ math.\ Statistics\\
           Sidlerstrasse 5, 3012 Bern, Switzerland
}

\date{Received: 25 Mai 2016 / Accepted: date}

\maketitle

\begin{abstract}
A variety of estimators for the parameters of the Generalized Pareto distribution, the approximating distribution for excesses over a high threshold, have been proposed, always assuming the underlying data to be independent. We recently proved  in \cite{Martig} that the likelihood moment estimators are consistent estimators for the parameters of the Generalized Pareto distribution for the case where the underlying data arises from a (stationary) linear process with heavy-tailed innovations. In this paper we derive the  bivariate asymptotic normality under some additional assumptions and give an explicit example on how to check these conditions by using asymptotic expansions.
\end{abstract}

\section{Fitting a Generalized Pareto distribution for a linear process with regularly varying tails}
\setcounter{equation}{0}
\renewcommand{\theequation}{\thesection.\arabic{equation}}
 We consider a \emph{(strictly) stationary} linear process
\begin{equation}
X_n = \sum_{j=0}^{\infty} c_j Z_{n-j} \tag{\text{A.}1} \label{A1}
\end{equation}
whose iid innovations $Z_n$ have a marginal distribution function $G_{Z}$ with regularly varying (i.e heavy) tails of index $- 1 / \gamma$, for $\gamma > 0$,  i.e.
\begin{equation}
1 - G_Z(z) \sim \pi_1 z^{-1/\gamma}L(z) \hspace{0.4cm} \text{and} \hspace{0.4cm} G_Z(-z)  \sim \pi_2 z^{-1/\gamma}L(z) \hspace{0.4cm} \text{as} \hspace{0.2cm} z \to \infty, \tag{\text{A.}2} \label{A2}
\end{equation}
for $\pi_1, \pi_2 \geq 0$, $\pi_1+ \pi_2 = 1$ with a slowly varying function $L(z)$,
where  $\{c_j\}_{j \geq 0} \in \mathbb{R}^{\infty}$. %satisfying $L(yz) \sim L(z)$ as $z \to \infty$ for any $y > 0$.
%Here we used the common notation $g(x) \sim f(x)$ as $x \to \infty$ to indicate $\lim_{x \to \infty} g(x)/f(x) = 1$.
Clearly, (\ref{A2}) implies $1 - G_{|Z|}(z) \sim z^{-1/\gamma}L(z)$, where $G_{|Z|}$ is the (marginal) distribution function of $|Z_n|$.
\vspace{0.25cm}
\newline Furthermore, for $\{c_j\}_{j \geq 0}$ being the sequence of coefficients in (\ref{A1}), let's assume that there is at least one $c_j \neq 0$ and there exist $A > 0$ and $u > 1$ such that
\begin{equation}
|c_j| < A u^{-j}, \hspace{0.5cm} j \in \mathbb{N}_0.  \tag{\text{A.}3} \label{A3}
\end{equation}
Here, it is worth to mention that every causal ARMA($p$,$q$) process of the form
\begin{equation*}
X_n - \phi_1 X_{n-1} - \ldots - \phi_p X_{n-p} = Z_n + \theta_1 Z_{n-1} + \ldots + \theta_q Z_{n-q},
\end{equation*}
$\phi_1, \ldots, \phi_p; \theta_1, \ldots, \theta_q \in \mathbb{R}$, has a representation as in (\ref{A1}) such that (\ref{A3}) automatically holds (\cite{Brockwell}, p. 85). From now on, (\ref{A1})-(\ref{A3}) will be considered as \textbf{Assumptions 1-3}.
\vspace{0.25cm}
\newline If we denote the marginal distribution function of $|X_n|$ by $\F_{|X|}$, then, with the help of Lemma 5.2 of \cite{Datta}, we directly conclude that the Assumptions 1-3 imply
\begin{equation}
\lim_{t \to \infty} \frac{1-\F_{|X|}(t)}{1 - \G_{|Z|}(t)}= \sum_{k=0}^{\infty} |c_k|^{1/\gamma} := ||c||. \label{Result 2.0.1}
\end{equation}
Under mild restrictions on the coefficients $c_j$ the tail behavior of $\F_{|X|}$ thus coincides with that of $\G_{|Z|}$ up to the constant $|| c ||$ and consequently, the marginal distribution of the time series $|X_n|$ has also a regularly varying tail of index $- 1 / \gamma$. According to \cite{Pickands}, Theorem 7, it then follows that there exists a positive function $\sigma^*(t)$ such that for $\F_{|X|,t}(x):= \Prob\big{(}|X_n|-t \leq x \, \big{|} |X_n| > t\big{)}$:
\vspace{0.2cm}
\begin{equation}
\lim_{t \to \infty} \: \sup_{x > 0} \: \left|\F_{|X|,t}(x) - \left[1- \left(1+ \frac{\gamma}{\sigma^* (t)}x \right)^{-1/\gamma} \right] \right| = 0. \label{domain4}
\end{equation}
\vspace{0.1cm}
\newline
The function $\Ha_{\gamma, \sigma^*}(x) = 1- \left(1+ \gamma x / \sigma^* (t) \right)^{-1/\gamma}$ is known as the Generalized Pareto distribution (GPD) with scale function $\sigma^* (t)$. (In fact, if $\gamma > 0$, the GPD simply equals to the Pareto distribution.) Relation (\ref{domain4}) thus tells us, that the random variable $|X_n|-t \, \big{|}|X_n|>t$, the excess above a high threshold $t$, has approximately a Generalized Pareto distribution. In practice, $t$ is replaced by the $(k + 1)$th largest observation $|X|_{n,n-k}$ for a sufficiently large $k << n$, i.e.\ $k=k(n)=o(n)$. Consequently, the target will be to estimate $\gamma$ and $\sigma^* (t)$ using the $k$ excess-values $|X|_{n,n} - |X|_{n,n-k}, \ldots, |X|_{n,n-k+1} - |X|_{n,n-k}$.
\vspace*{0.25cm}
\newline A variety of estimators have been proposed for the two parameters $\gamma$ and $\sigma^*  = \sigma^* (t)$ such as the maximum likelihood estimators (\cite{Smith}), the moment and probability weighted moment estimators (\cite{Hosking}), the likelihood moment estimators (\cite{Zhang}) or the goodness-of-fit estimators (\cite{Huesler}), to mention only a few. A nice r\'{e}sum\'{e} is also available in \cite{deHaan}.
\vspace{0.25cm}
\newline We showed in \cite{Martig} that the likelihood moment estimators $(\hat{\gamma}_{LME}$,$\hat{\sigma}_{LME})^T$, the solutions of the following system of equations for $\gamma$ and $\sigma^*$:
\vspace{0.2cm}
\begin{eqnarray}
  \frac{1}{k} \sum_{i=0}^{k-1} \log \left(1 + \frac{\gamma}{\sigma^*} \left(|X|_{n,n-i} - |X|_{n,n-k} \right)  \right)  & = & \gamma \label{LME1} \\[0.1cm]
 \frac{1}{k} \sum_{i=0}^{k-1} \left(1 + \frac{\gamma}{\sigma^*} \left(|X|_{n,n-i} - |X|_{n,n-k} \right)  \right)^{r/\gamma} & = & (1-r)^{-1}, \label{LME2}
\end{eqnarray}
\vspace{-0.1cm} \newline
are consistent estimators for the parameters of the Generalized Pareto distribution if (\ref{A1})-(\ref{A3}) hold and $r < 0$. In this paper, we prove the asymptotic bivariate normality of $(\hat{\gamma}_{LME}$,$\hat{\sigma}_{LME})^T$ under some additional assumptions, see Theorem \ref{Theo:Asnorm} and Corollary \ref{Cor:Asnorm} in Section 2,
$$\sqrt{k} \left( \begin{array}{c} \displaystyle \hat{\gamma}_{LME} - \gamma \\[0.25cm] \displaystyle \frac{\hat{\sigma}_{LME}}{\sigma(n/k)} - 1 \end{array} \right) \overset{D}{\to} N\left( \boldsymbol{0},L \Sigma L^T\right).
$$

\section{Asymptotic normality of $(\hat{\gamma}_{LME}$,$\hat{\sigma}_{LME})^T$}
\setcounter{equation}{0}
In this section we prove asymptotic bivariate normality of $(\hat{\gamma}_{LME}$,$\hat{\sigma}_{LME})^T$. Thereby, we will use and also extend some of the results about tail array sums of \cite{Rootzen1} and \cite{Rootzen2}. For the calculation of the covariance matrix, the methods of \cite{Resnick4} are applied.
\vspace{0.25cm}
\newline For our result we need further assumptions. Throughout, the sequence $k = k(n), k =o(n)$ refers to the $(k + 1)$th largest observation $|X|_{n,n-k}$  and for simplicity, we will assume that there exist no ties and that $|X|_{n,n-k}$ is uniquely defined (notice that -- under some weak additional conditions -- all results surely continue to hold in the presence of ties, see Remark 4.4 in \cite{Rootzen1}). Also, we assume $r$ in (\ref{LME2}) to be negative and finally, for $b_{|X|}(t) := (1/(1-\F_{|X|}))^{\leftarrow}(t) := \inf\{y : 1/(1-\F_{|X|}(y)) \geq t  \}$ being the $1 - 1/t$-quantile of $\F_{|X|}$, let's define
\begin{equation}
\sigma(t) := \sigma^*(b_{|X|}(t)) \label{aissig}
\end{equation}
so that for any $x > 0$ as $t \to \infty$ (see \cite{Martig}):
\begin{equation*}
\frac{b_{|X|}(tx)-b_{|X|}(t)}{\sigma(t)} \to \frac{x^{\gamma} - 1 }{\gamma}.
\end{equation*}
\begin{definition} Let $\mathfrak{B}_{i,j}$ be the $\sigma$-field $\sigma(X_h)_{i \leq h \leq j}$ generated by a sequence $X_i, X_{i+1},$ $\ldots,X_j$.  Then for fixed $n, l \in \mathbb{N} \, \diagdown \, \{0\}$ with $l < n$ we define
      \begin{equation*}
      \alpha_{n,l} := \sup(|\Prob(A \cap B) - \Prob(A)\Prob(B)|: A \in \mathfrak{B}_{1,h}, B \in \mathfrak{B}_{h+l,n}, 1 \leq h \leq n-l),
      \end{equation*}
       and we say that $\{X_j\}$ is $(\alpha_{n,l_n},l_n)$-strongly mixing if $\alpha_{n,l_n} \to 0$ for some $l_n\to \infty$ with $ l_n=o(n)$.
\end{definition}
\begin{definition} Let $b_{|X|}(t)$ and $\sigma(t)$ be defined as in (\ref{aissig}), then we say $b_{|X|}$ is second-order extended regularly varying, 2ERV$(\rho, A)$ for short, if for $x > 0$, $\rho \leq 0$, $\gamma \neq -\rho$, there exists a positive (or negative) function $A(t)$ of constant sign near infinity, $A(t) \to 0$ as $t \to \infty$, such that
\begin{equation*}
\lim_{t \to \infty} \left. \frac{\displaystyle \frac{b_{|X|}(tx) - b_{|X|}(t)}{\sigma(t)} - \frac{x^\gamma - 1}{\gamma}}{A(t)} \right. = \frac{1}{\rho} \left(\frac {x^{\gamma + \rho} - 1}{\gamma + \rho} -  \frac{x^\gamma - 1}{\gamma } \right)
\end{equation*}
and
\begin{equation*}
\lim_{n,k,n/k \to \infty}  \sqrt{k} A(n/k) = 0.
\end{equation*}
\end{definition}
\begin{definition}
$1 - \F_{|X|}$ is said to be second-order regularly varying, 2RV$(\rho', A^*)$ for short, if for $x > 0$, $\rho' \leq 0$, $\gamma \neq -\rho'$, there exists a positive (or negative) function $A^*(t)$ of constant sign near infinity, $A^*(t) \to 0$ as $t \to \infty$, such that
\begin{equation*}
\lim_{t \to \infty} \left. \frac{\displaystyle \frac{1-\F_{|X|}(tx)}{1-\F_{|X|}(t)} - x^{-1/\gamma}}{A^*(t)} \right. = x^{-1/\gamma} \, \frac{x^{\rho'} -1}{\rho'}
\end{equation*}
and
\begin{equation*}
\lim_{n,k,n/k \to \infty} \sqrt{k} A^*(b_{|X|}(n/k)) \to 0.
\end{equation*}
\end{definition}
\setlength{\parindent}{0mm}
The following assumptions will be of prime importance throughout this section:
\vspace{0.1cm}
\newline Let $\{c_j \}_{j \geq 0}$ be the sequence of coefficients in (A.1).
 %and $ \mathbf{1} \{c_i, c_{i + j} \neq  0\} := \mathbf{1}\{ c_i \neq 0 \} \cdot \mathbf{1}\{ c_{i + j}\neq 0 \}$.
Then we assume that
\begin{equation}
\sum_{j=1}^{\infty} \sum_{i=0}^{\infty} \left( |c_i| \wedge |c_{i+j}| \right)^{1/\gamma} \log \left(  \frac{ |c_i| \vee |c_{i+j}|}{|c_i| \wedge |c_{i+j}| }\right)  \mathbf{1} \{c_i c_{i + j} \neq  0\}  < \infty. \tag{\text{A.}4} \label{A4}
\end{equation}
\vspace{0.1cm}
\newline
 The sequence of thresholds $u_n$ is chosen such that $n(1-\F_{|X|}(\exp(u_n))) = k + 1$.
\newline ${}$ \hspace{10.8cm} (A.5)
\vspace{0.1cm}
\newline
$\{X_j\}$ is $(\alpha_{n,l_n},l_n)$-strongly mixing. \hspace{5.5cm} (A.6)
\vspace{0.1cm}
\newline
A sequence of integers $r_n$ with $l_n < r_n \leq n$ is chosen such that $kr_n = o(n)$ and for $m_n:= [n/r_n]$ we have $m_n(\alpha_{n,l_n} + l_n/n) \to 0$ as $n \to \infty$. \hspace{2cm} (A.7)
\vspace{0.1cm}
\newline
$\F_{|X|}$ has a derivative $\F'_{|X|}$ in $(t_0,\infty)$ for some $t_0 > 0$ and we assume that
\begin{equation*}
\lim_{t \to \infty}\frac{t \F'_{|X| }(t)}{1-\F_{|X|}(t)} = \frac{1}{\gamma} \tag{\text{A.}8}. \label{A8}
\end{equation*}
\vspace{0.1cm}
\newline
$b_{|X|} \in $ 2ERV$(\rho, A)$ and $1 - \F_{|X|} \in $ RV$(\rho', A^*)$. \hspace{4cm} (A.9)
\vspace{0.1cm}
\newline
The properties (A.4)-(A.9) will be called \textbf{Assumptions 4-9}. Notice that the sequence $r_n$ in (A.7) always exists (e.g. $r_n = \lceil \max(n \, \alpha_{n,l_n}^{1/2}, n^{3/2} \, l_n^{-1/2}) \rceil $). The sequence $m_n$ is called the standard sequence which decomposes $(0,n]$  into intervals $J_i := ((i-1)r_n, ir_n], 1 \leq i \leq m_n$ and a last interval $J_{m_n + 1} := (m_n r_n, n]$ termed the standard partition. Finally, one can show that if $0 <-\rho < \gamma$, then $b_{|X|} \in $ 2ERV$(\rho, A)$ implies $1 - \F_{|X|} \in $ RV$(\rho', A^*)$ with $\rho' = \rho/\gamma$ and $A^*(t)=A(1/(1-\F_{|X|}(t)))/(\gamma^2 (1 + \rho'))$. In case $0 < \gamma <-\rho$ the additional assumption $\lim_{t \to \infty} \gamma \, b_{|X|}(t) - \sigma(t) = 0$ is needed (\cite{Neves}, Theorem 1, and \cite{deHaan}, Remark 2.3.10).
\vspace{0.1cm}
\newline Now, since $1 - \F_{|X|}$ is regularly varying, it directly follows by (A.5) that for any $x >0$:
\begin{equation}\frac{n}{k}\, (1-\F_{|X|}(x\, \exp(u_n)) )\sim \frac{1-\F_{|X|}(x \, \exp(u_n))}{1-\F_{|X|}( \exp(u_n))} \sim x^{-1/\gamma}\label{u_nsim_b}
\end{equation}
as $n, k, n/k \to \infty$, which can be extended by an application of Lemma 5.1 and Lemma 5.2 of \cite{Datta} to
\begin{eqnarray}
 &\hspace*{-15mm} \frac{n}{k} \Prob\left( |X_1|> x \exp(u_n), |X_{j+1}|> y \exp(u_n)\right) \nonumber \\
  &\sim \frac{1}{||c|| } \sum_{k = 0}^{\infty}   \left(|c_k|^{1/\gamma} x^{-1/\gamma} \wedge |c_{j+k}|^{1/\gamma} y^{-1/\gamma}  \right),  \label{u_nsim_b2}
\end{eqnarray}
\vspace{-0.1cm}
\newline
locally uniformly in $x,y \in (0,\infty)$, where any $j \in \mathbb{N}_0$.
\vspace{0.25cm}
\newline Next consider the two functions
\begin{equation*}
\phi_1(x) := x \mathbf{1}\{x > 0 \} \hspace{0.25cm} \text{and} \hspace{0.25cm} \phi_2(x) := (1 - \exp(rx/\gamma))\mathbf{1}\{x > 0 \},
\end{equation*}
where $r < 0$. As we will see later on, there exists a clear connection between the tail array sums involved in the likelihood moment estimation (i.e. (\ref{LME1}) and (\ref{LME2})) and the ones implied by the functions $\phi_1$ and $\phi_2$ when considering the transformed sequence $\{\log(|X_j|)\}$. Furthermore, the properties in (\ref{u_nsim_b}) and (\ref{u_nsim_b2}) are  such that they allow us to use and slightly modify the results of Resnick \& St\^aric\^a in \cite{Resnick4}. After some repetition, it is evident that their convergences stated in Lemma 2.1 and Remark 2.1 as well as their methods applied in the respective proofs continue to hold with $b_{|X|}(n/k)$ replaced by our $\exp(u_n)$ and this allows us to calculate the means, the variances and the covariances of tail array block-sums based on $\phi_i(\log(|X_j|/\exp(u_n))) = \phi_i(\log(|X_j|)- u_n), i = 1,2$.
\vspace{0.25cm}
\newline
Throughout, we will use the following notations for $Y_j := \log(|X_j|/\exp(u_n))$, $i = 1,2$ and $r < 0$:
\vspace{0.25cm}
\begin{eqnarray*}
\beta'_1 &:=& \frac{\gamma}{(\gamma + 1)}, \hspace{0.25cm} \beta'_2 := -\frac{r}{(1-r + \gamma)} \hspace{0.25cm} \text{and} \hspace{0.25cm}  \tilde{\phi}_i(x) := \phi_i(x) - \beta'_{i} \mathbf{1}\{x > 0 \}, \\[0.25cm]
\left(\tau_n^{(i)}\right)^2 &:=& \frac{n}{r_n} \Var\left(\sum_{j = 1}^{r_n} \phi_i(Y_j) \right), \hspace{0.25cm} \left(\tau_n^{(I)}\right)^2 := \frac{n}{r_n} \Var\left(\sum_{j = 1}^{r_n} \mathbf{1}\{Y_j > 0\} \right), \\[0.1cm]
\tau_n^{(1,2)} &:=& \frac{n}{r_n} \Cov \left( \sum_{j = 1}^{r_n} \phi_1(Y_j), \sum_{j = 1}^{r_n} \phi_2(Y_j)  \right),
\\[0.1cm]
\tau_n^{(i,I)} &:=& \frac{n}{r_n} \Cov \left( \sum_{j = 1}^{r_n} \phi_i(Y_j), \sum_{j = 1}^{r_n} \mathbf{1}\{Y_j > 0\}  \right),
\\[0.1cm]
\left(\tilde{\tau}_n^{(i)}\right)^2 &:=& \frac{n}{r_n} \Var\left(\sum_{j = 1}^{r_n} \tilde{\phi}_i(Y_j) \right), \hspace{0.25cm} \tilde{\tau}_n^{(1,2)} := \frac{n}{r_n} \Cov \left( \sum_{j = 1}^{r_n} \tilde{\phi}_1(Y_j) , \sum_{j = 1}^{r_n} \tilde{\phi}_2(Y_j) \right).
\end{eqnarray*}
Finally, for $\{c_j \}_{j \geq 0}$ being again the sequence of coefficients in (A.1), we define
\stepcounter{equation}
\begin{align*}
\varphi_1 &:= \frac{1}{||c||}\sum_{j=1}^{\infty} \sum_{k=0}^{\infty} \left( |c_k| \wedge |c_{j+k}| \right)^{1/\gamma} ,\\[0.2cm]
\addtocounter{equation}{1}
\varphi_2 &:= \frac{1}{||c||} \sum_{j=1}^{\infty} \sum_{k=0}^{\infty} \frac{ \left( |c_k| \vee |c_{j+k}| \right)^{r/\gamma}}{\left( |c_k| \wedge |c_{j+k}| \right)^{(r-1)/\gamma}}  \mathbf{1} \{c_k, c_{j + k} \neq  0\}  , \\[0.2cm]
\varphi_3 &:= \frac{1}{||c||} \sum_{j=1}^{\infty} \sum_{k=0}^{\infty} \left( |c_k| \wedge |c_{j+k}| \right)^{1/\gamma} \log \left(  \frac{ |c_k| \vee |c_{j+k}|}{|c_k| \wedge |c_{j+k}| }\right)  \mathbf{1} \{c_k, c_{j + k} \neq  0\}.
\end{align*}
\begin{res} \label{Res:Variances} Let the Assumptions 1-5 hold. Then for a sequence of integers $r_n \to \infty$ satisfying $kr_n = o(n)$ we have as $n,k,n/k \to \infty$:
\vspace{0.25cm}
\begin{align}
&\frac{\left(\tau_n^{(1)}\right)^2}{k} \to 2\gamma \left(\gamma + 2 \gamma \varphi_1 + \varphi_3 \right), \label{2.2.9.1} \\[0.1cm]
&\frac{\left(\tau_n^{(2)}\right)^2}{k} \to \frac{-2r (-r + (1-2r)\varphi_1 - \varphi_2)}{(1-r)(1-2r)}, \label{2.2.9.3} \\[0.1cm]
&\frac{\left(\tau_n^{(I)}\right)^2}{k} \to 1 + 2\varphi_1, \label{2.2.9.I} \\[0.1cm]
&\frac{\tau_n^{(1,2)}}{k} \to  \frac{ \left(-\gamma r (2-r) + (2r^2 -4r + 1)\gamma \varphi_1 -\gamma \varphi_2
- r(1-r)\varphi_3  \right)}{(1-r)^2}, \label{2.2.9.ICov2} \\[0.1cm]
&\frac{\tau_n^{(1,I)}}{k} \to \gamma + 2\gamma\varphi_1 + \varphi_3, \label{2.2.9.ICov1} \\[0.1cm]
&\frac{\tau_n^{(2,I)}}{k} \to  \frac{\left(-r + (1-2r)\varphi_1 - \varphi_2  \right) }{(1-r)}, \label{2.2.9.5} \\[0.1cm]
&\frac{\left(\tilde{\tau}_n^{(1)}\right)^2}{k} \to  \frac{(1+2\varphi_1)\gamma^2 \left( 2 \gamma^2 + 2\gamma + 1 \right) + 2\gamma^2 (\gamma + 1) \varphi_3}{(\gamma + 1)^2} := \kappa_1(\gamma,r), \label{2.2.9.2} \\[0.1cm]
&\frac{\left(\tilde{\tau}_n^{(2)}\right)^2}{k} \to \frac{-2\gamma r(\gamma + 1)(-r + \varphi_1(1-2r)-\varphi_2) +r^2(1-r)(1+2\varphi_2)}{(1-r)(1-2r)(1-r+\gamma)^2}:=  \kappa_2 (\gamma,r), \label{2.2.9.4} \\[0.1cm]
&\frac{\tilde{\tau}_n^{(1,2)}}{k} \to - \frac{\gamma r ((2-r)(\gamma^2+\gamma + 1) -1 ) }{(1-r)^2 (1 + \gamma)(1-r+\gamma)} \nonumber \\[0.1cm] &  \hspace{1.45cm} - \frac{\gamma (
2r(2-r)(\gamma^2+ \gamma + 1) - (\gamma^2 + \gamma + 3r -r^2))}{(1-r)^2 (1 + \gamma)(1-r+\gamma)} \cdot \varphi_1 \nonumber \\[0.1cm] &  \hspace{1.45cm} - \frac{
\gamma (\gamma + r) }{(1-r)^2 (1+\gamma)} \cdot \varphi_2 - \frac{\gamma r}{(1-r)(1-r + \gamma)} \cdot \varphi_3 :=\kappa_3(\gamma,r). \label{2.2.9.6}
\end{align}
\end{res}
\begin{proof} By virtue of (\ref{u_nsim_b}) and (\ref{u_nsim_b2}), the convergences (\ref{2.2.9.1})-(\ref{2.2.9.6}) directly follow after a small repetition of the proofs of Lemma 2.1 and Remark 2.1 in \cite{Resnick4} with $b_{|X|}(n/k)$ replaced by our $\exp(u_n)$. \qed
\end{proof}
The next result shows the asymptotic behavior of \newline $S_n^{(i)}(x) := (\sum_{j = 1}^n \mathbf{1} \{\log(|X_j|) > x \} - n(1-\F_{|X|}(\exp(x))))/\tilde{\tau}_n^{(i)}, i = 1,2.$
\begin{lem} \label{Lem:Sn}  Suppose the Assumptions 1-8 hold. Additionally, let
\begin{equation}
  \frac{r_n}{\sqrt{k}} \to 0 \label{w_n3}
\end{equation}
as $n,k,n/k \to \infty$. Then for $i = 1,2$:
\begin{align}
& S_n^{(i)}(u_n) \overset{D}{\to} N(0,(1 + 2 \varphi_1)/\kappa_i(\gamma,r)), \label{2.2.10.1} \\[0.1cm]
& S_n^{(i)}(\log(|X|_{n,n-k})) - S_n^{(i)}(u_n)  \overset{P}{\to} 0, \label{2.2.10.3} \\[0.1cm]
& k(\log(|X|_{n,n-k}) - u_n)/\tilde{\tau}_n^{(i)} - \gamma  S_n^{(i)}(u_n)   \overset{P}{\to} 0. \label{2.2.10.4}
\end{align}
\end{lem}
\begin{proof} Let $i = 1$ and recall $Y_j = \log(|X_j|/\exp(u_n))$. By definition
\begin{equation*}
 S_n^{(1)}(u_n) = \frac{\sum_{j=1}^n \left[ \mathbf{1}\{Y_j > 0 \} - \E(\mathbf{1}\{Y_1 > 0 \} ) \right] }{\tau^{(I)}_n} \cdot \frac{\tau^{(I)}_n}{\tilde{\tau}_n^{(1)}} := \text{I}_1  \cdot \text{I}_2 .
\end{equation*}
Using all our assumptions, it is straightforward that all conditions of Theorem 6.2 of \cite{Rootzen2} hold for the transformed sequence $\{\log(|X_j|)\}$ and tail array function $\phi(x) = \mathbf{1}\{x > 0\}$ and this simply means that $\text{I}_1 \overset{D}{\to} N(0,1)$ as $n,k,n/k \to \infty$. With the help of Result \ref{Res:Variances} it is then easy to see that $\text{I}_2 \to \sqrt{(1 + 2 \varphi_1)/\kappa_1(\gamma,r)}$ as $n,k,n/k \to \infty$. This shows (\ref{2.2.10.1}) for $i=1$.
\newline To show (\ref{2.2.10.3}) and (\ref{2.2.10.4}), it suffices to check whether all conditions of Lemma 4.2 of \cite{Rootzen1} are again satisfied for $\{\log(|X_j|)\}$ and the task then reduces to prove that for \emph{any} non-random sequence $z_n$ with $k(z_n - u_n)/\tilde{\tau}_n^{(1)}$ bounded we have $S^{(1)}_n(z_n) - S^{(1)}_n(u_n) \overset{P}{\to} 0$ as $n,k,n/k \to \infty$. By the very same paper (pp. 21-23), this condition is satisfied if for $v_n := \exp(z_n)$, $I_n := [u_n,\log(v_n))$ in case $\log(v_n) > u_n$ and $I_n := [\log(v_n),u_n)$ in case $\log(v_n) < u_n$:
\begin{equation}
\frac{n}{r_n \left(\tilde{\tau}_n^{(1)}\right)^2} \Var \left(\sum_{j=1}^{r_n} \mathbf{1} \{\log(|X_j|) \in I_n\} \right) \to 0 \label{Reshelp}
\end{equation}
whenever $v_n/\exp(u_n) \to 1$ as $n,k,n/k \to \infty$. But going along the steps of the proof of Lemma 2.1 and Remark 2.1 of \cite{Resnick4} (or see the proof of Lemma \ref{Lem:NegCon} below) the convergence in (\ref{Reshelp}) is obvious. \\
For the case $i = 2$ analogous considerations  show the statements. \qed
\end{proof}
The next result shows the limit behavior of the difference of our selected tail array sums when two different thresholds are chosen: the random value $\log(|X|_{n,n-k})$ in the first case and $u_n$ in the second case.
\begin{res} \label{Res:Firstdiff} Let Assumptions 1-8 and (\ref{w_n3}) hold. Recall that $Y_j = \log(|X_j|) - u_n$ and define $Y'_j := \log(|X_j|) - \log(|X|_{n,n-k})$. Then for $r < 0$:
\begin{eqnarray}
&& \hspace{-1cm} \frac{1}{\tilde{\tau}_n^{(1)}}  \sum_{j = 1}^{n} \left[ \phi_1(Y'_j)  -  \phi_1(Y_j) \right] + \gamma S^{(1)}_n(u_n) \overset{P}{\to} 0, \label{Hill_Hill.alt} \\[0.1cm]
&&  \hspace{-1cm} \frac{1}{\tilde{\tau}_n^{(2)}}   \sum_{j = 1}^{n} \left[ \phi_2(Y'_j)  -  \phi_2(Y_j) \right] -  \frac{r S_n^{(2)}(u_n)}{  (1-r)}  \overset{P}{\to} 0, \label{2nd_2nd.alt}
\end{eqnarray}
as $n,k,n/k \to \infty$.
\end{res}
\begin{proof} The proof of (\ref{Hill_Hill.alt}) was originally given by \cite{Rootzen1}, pp. 15-19. Thus we will lay our focus on (\ref{2nd_2nd.alt}), which can be proved in a very similar way. To begin, let's retain from the aforementioned paper that as  $n,k,n/k \to \infty$:
\begin{equation}
R_n^{(2)} := \frac{1}{\tilde{\tau}_n^{(2)}} \sum_{j=1}^{n} (\log(|X_j|) - u_n) \mathbf{1}\left\{\log(|X|_{n,n-k}) \geq \log(|X_j|) > u_n \right\} \overset{P}{\to} 0, \label{reminder1}
\end{equation}
and similarly,
\begin{equation}
T_n^{(2)} := \frac{1}{\tilde{\tau}_n^{(2)}} \sum_{j=1}^{n} (\log(|X_j|) -  \log(|X|_{n,n-k})) \mathbf{1}\left\{u_n \geq \log(|X_j|) > \log(|X|_{n,n-k}) \right\} \overset{P}{\to} 0. \label{reminder2}
\end{equation}
Also, using standard arguments (\cite{Martig}, Lemma 1, and \cite{Resnick3}, pp. 80-85), one quickly shows that
\begin{equation}
\frac{1}{k} \sum_{j=0}^{k-1} \left( \frac{|X|_{n,n-j}}{|X|_{n,n-k}} \right)^{r/\gamma}  \overset{P}{\to} \frac{1}{(1-r)} \label{(1-r)},
\end{equation}
$n,k,n/k \to \infty$.
\vspace{0.25cm}
\newline Now, let's consider the set $\{|X|_{n,n-k} > \exp(u_n)\}$. Usually, this can be done by multiplying $\mathbf{1}\{ |X|_{n,n-k} > \exp(u_n)\}$ throughout, but for simplicity, we will not denote it and just assume $|X|_{n,n-k} > \exp(u_n)$ in the computations below (cf. \cite{Rootzen1}). Then
\begin{eqnarray*}
&&  \frac{1}{\tilde{\tau}_n^{(2)}}   \sum_{j = 1}^{n} \left[ \phi_2(Y'_j)  -  \phi_2(Y_j) \right]  \\[0.1cm]
&=& \frac{1}{\tilde{\tau}_n^{(2)}} \sum_{j = 0}^{k-1} \left[  \left( \frac{|X|_{n,n-j}}{\exp(u_n)} \right)^{r/\gamma} -   \left( \frac{|X|_{n,n-j}}{|X|_{n,n-k}} \right)^{r/\gamma}   \right]  \\[0.1cm]
&& - \frac{1}{\tilde{\tau}_n^{(2)}}  \sum_{j = 1}^n  \left( 1- \left( \frac{|X_j|}{\exp(u_n)} \right)^{r/\gamma} \right)  \mathbf{1}\left\{ |X_{n,n-k}| \geq |X_j| > \exp(u_n) \right\}  := \text{I}_1+ \text{I}_2.
\end{eqnarray*}
Next, since $|X|_{n,n-k}/ \exp(u_n) \overset{P}{\to} 1$ as $n,k,n/k \to \infty$ by (\ref{2.2.10.4}) and $z^{r/\gamma} - 1 \sim r\log(z)/\gamma$ as $z \to 1$, an application of the continuous mapping theorem yields
\begin{eqnarray*}
\text{I}_1 &=& \frac{1}{\tilde{\tau}_n^{(2)}} \sum_{j = 0}^{k-1}  \left( \frac{|X|_{n,n-j}}{|X|_{n,n-k}} \right)^{r/\gamma}  \left( \left( \frac{|X|_{n,n-k}}{\exp(u_n)} \right)^{r/\gamma} -   1 \right) \\[0.1cm]
&=& \frac{r}{\gamma} \frac{k}{\tilde{\tau}_n^{(2)}}   \left( \log(|X|_{n,n-k}) - u_n  \right)  (1 + o_P(1)) \frac{1}{k}  \sum_{j = 0}^{k-1}  \left( \frac{|X|_{n,n-j}}{|X|_{n,n-k}} \right)^{r/\gamma} \\[0.1cm]
&=& \frac{r}{\gamma (1-r)} \frac{k}{\tilde{\tau}_n^{(2)}}   \left( \log(|X|_{n,n-k}) - u_n  \right) (1 + o_P(1)) =  \frac{r}{  (1-r)}  S_n^{(2)}(u_n)+ o_P(1),
\end{eqnarray*}
where (\ref{2.2.10.1}), (\ref{2.2.10.4}) and (\ref{(1-r)}) were used in the last line.
\newline Finally, a simple application of the mean value theorem shows that $|\text{I}_2| \leq   - r R_n^{(2)}/\gamma$, but this simply means $ \text{I}_2 \overset{P}{\to} 0$ as $n,k,n/k \to \infty$ by (\ref{reminder1}).
\newline The same result holds on the set $\{ |X|_{n,n-k} < \exp(u_n) \}$ using similar arguments as above, (\ref{reminder2}) and by noticing that $\frac{1}{k}\sum_{j = 1}^n \mathbf{1} \{|X_j|> \exp(u_n) \} \overset{P}{\to} 1$ as $n,k,n/k \to \infty$. \qed \end{proof}
The next step will be to give the link between the tail array sums related to the likelihood moment estimators and our tail array sums selected at the very beginning of the section. The crucial point here is to make use of both the second-order and the second-order extended regular variation introduced in Assumption 9. With the help of Lemma B.3.16 of \cite{deHaan}, we quickly conclude that if $b_{|X|} \in $ 2ERV$(\rho, A)$, then
\begin{equation}
\frac{\gamma b_{|X|}(n/k)}{ \sigma(n/k) } = 1 + o(1/\sqrt{k}) \label{A(t)_1}
\end{equation}
as $n,k,n/k \to \infty$. Furthermore, going along the steps of \cite{Hsing}, pp. 1552-1554, we also deduce that if $1 - \F_{|X|} \in $ RV$(\rho', A^*)$, then $1-\F_{|X|}(b_{|X|}(n/k)) = k(1+o(1/\sqrt{k}))/n$. But since we chose $u_n$ such that $1-\F_{|X|}(\exp(u_n)) $ $= (k+1)/n$, it follows that
\begin{equation}
\frac{1-\F_{|X|}(\exp(u_n))}{1-\F_{|X|}(b_{|X|}(n/k))} = 1 + o(1/\sqrt{k}) \label{tailratio_k}
\end{equation}
and thus $A^*(b_{|X|}(n/k)) \sim A^*(\exp(u_n))$ and finally
\begin{equation}
\sqrt{k} \left( \frac{n}{k} \E\left( \phi_i(Y_1) \right) - \beta_i \right) \to 0, \label{Expconv}
\end{equation}
for $i = 1,2$, $\beta_1 := \gamma$  and $  \beta_2 := -r/(1-r)$. In conclusion, considering Assumption 8,
\begin{eqnarray*}
\frac{1-\F_{|X|}(\exp(u_n))}{1-\F_{|X|}(b_{|X|}(n/k))} - 1  &\sim&  \log\left(\frac{1-\F_{|X|}(\exp(u_n))}{1-\F_{|X|}(b_{|X|}(n/k))} \right) \nonumber \\[0.2cm] &=& \log (1-\F_{|X|}(\exp(u_n))) - \log( 1-\F_{|X|}(b_{|X|}(n/k))) \nonumber \\[0.2cm] &=& - \int_1^{\exp(u_n)/b_{|X|}(n/k)} \frac{ b_{|X|}(n/k) s \F'_{|X|}(s b_{|X|}(n/k))}{1-\F_{|X|}(s  b_{|X|}(n/k))} \, \frac{ds}{s} \nonumber
\\[0.2cm] &\sim& - \frac{1}{\gamma} \int_1^{\exp(u_n)/b_{|X|}(n/k)} \frac{ds}{s}\nonumber  = - \frac{1}{\gamma}\log(\exp(u_n)/b_{|X|}(n/k)) \\[0.25cm] &\sim& \frac{1}{\gamma} \left( \frac{b_{|X|}(n/k)}{\exp(u_n)} -1 \right),
\end{eqnarray*}
which simply means $b_{|X|}(n/k)/\exp(u_n) = 1 + o(1/\sqrt{k})$ by virtue of (\ref{tailratio_k}) and thus, using (\ref{A(t)_1}),
\begin{equation}
\frac{\gamma \exp(u_n)}{ \sigma(n/k) } = 1 + o(1/\sqrt{k}) \label{A(t)_2},
\end{equation}
 as $n,k,n/k \to \infty$.
\begin{res} \label{Res:Seconddiff} Let the Assumptions 1-9 and (\ref{w_n3}) hold. Recall that $Y'_j = \log(|X_j|) - \log(|X|_{n,n-k})$ and define $Y''_j := |X|_{n,n-j} - |X|_{n,n-k}$. Then:
\begin{eqnarray}
&&  \frac{1}{\tilde{\tau}_n^{(1)}}  \sum_{j = 0}^{k-1} \left[ \log \left( 1 + \frac{\gamma Y''_j}{\sigma(n/k)} \right) -   \phi_1(Y'_j ) \right]   -  \frac{\gamma^2 S_n^{(1)}(u_n)}{(\gamma + 1)}  \overset{P}{\to} 0, \label{Hill_Hill.alt.2} \\[0.1cm]
&&  \frac{1}{\tilde{\tau}_n^{(2)}}  \sum_{j = 0}^{k-1} \left[ \left( 1 + \frac{\gamma  Y''_j }{\sigma(n/k)} \right)^{r/\gamma} -(1 -  \phi_2(Y'_j)) \right]    -  \frac{\gamma r S_n^{(2)}(u_n)}{  (1-r)(1-r+\gamma)}  \overset{P}{\to} 0 \nonumber , \\ \label{2nd_2nd.alt.2}
\end{eqnarray}
as $n,k,n/k \to \infty$.
\end{res}
\begin{proof}
With the help of \cite{Martig}, Lemma 1, and (\ref{A(t)_1}) it is easy to see that for any $j < k$ we always have \vspace{-0.1cm}
\begin{equation}
\frac{|X|_{n,n-k}}{|X|_{n,n-j}} = O_P(1) \: \: \: \: \: \text{and} \: \: \: \: \: \frac{\gamma \, |X|_{n,n-k}}{\sigma(n/k)} \overset{P}{\to} 1 \label{twoprelim}
\end{equation}
as $n,k,n/k \to \infty.$
To prove (\ref{Hill_Hill.alt.2}), notice that
\vspace{-0.2cm}
\begin{eqnarray*}
&& \frac{1}{\tilde{\tau}_n^{(1)}}  \sum_{j = 0}^{k-1} \left[ \log \left( 1 + \frac{\gamma Y''_j}{\sigma(n/k)} \right)  -  \phi_1(Y'_j) \right] \\[0.1cm]
&=& \frac{1}{\tilde{\tau}_n^{(1)}}  \sum_{j = 0}^{k-1}   \log \left( \frac{|X|_{n,n-k}}{|X|_{n,n-j}} \left(1 -  \frac{\gamma \, |X|_{n,n-k}}{\sigma(n/k)}  \right)  +  \frac{\gamma \, |X|_{n,n-k}}{\sigma(n/k)}  \right) := \frac{1}{\tilde{\tau}_n^{(1)}}  \sum_{j = 0}^{k-1}   \log (A_{j,k,n}).
\end{eqnarray*}
Now, due to (\ref{twoprelim}), $A_{j,k,n} = 1 + o_P(1)$ uniformly for  $j<k$ as $n,k,n/k \to \infty$ and hence,
using the fact that $\log(z) \sim (z-1)$ as $z \to 1$, (\ref{2.2.9.2}), (\ref{2.2.10.4}), (\ref{(1-r)}), (\ref{A(t)_2}) and the continuous mapping theorem yield
\vspace{-0.2cm}
\begin{eqnarray*}
&& \frac{1}{\tilde{\tau}_n^{(1)}}  \sum_{j = 0}^{k-1}   \log (A_{j,k,n}) \\[0.1cm]
&=& \frac{1}{\tilde{\tau}_n^{(1)}}  \sum_{j = 0}^{k-1}  \left[ \frac{|X|_{n,n-k}}{|X|_{n,n-j}} \left(1 -  \frac{\gamma \, |X|_{n,n-k}}{\sigma(n/k)}  \right)  +  \frac{\gamma \, |X|_{n,n-k}}{\sigma(n/k)}   - 1 \right] (1+o_P(1)) \\[0.1cm]
&=& \frac{k}{\tilde{\tau}_n^{(1)}}  \left(\frac{\gamma \, |X|_{n,n-k}}{\sigma(n/k)} -1  \right)  \frac{1}{k} \sum_{j = 0}^{k-1}  \left( 1 - \frac{|X|_{n,n-k}}{|X|_{n,n-j}} \right) (1+o_P(1)) \\[0.1cm]
  &=& \frac{\gamma}{(\gamma + 1)} \frac{\sqrt{k}}{\tilde{\tau}_n^{(1)}} \left[ \sqrt{k} \: \frac{\gamma  \exp(u_n) }{ \sigma(n/k) } \left(\frac{|X|_{n,n-k}}{ \exp(u_n)}  - 1 \right) +  \sqrt{k} \left(\frac{\gamma  \exp(u_n)}{ \sigma(n/k) }   -1 \right)  \right]   (1+o_P(1)) \\[0.1cm]
&=& \frac{\gamma}{(\gamma + 1)} \frac{k}{\tilde{\tau}_n^{(1)}} \left(\frac{|X|_{n,n-k}}{ \exp(u_n)}  - 1 \right)   (1+o_P(1)) \\[0.1cm]
&=&  \frac{\gamma}{(\gamma + 1)} \frac{k}{\tilde{\tau}_n^{(1)}} \left(\log(|X|_{n,n-k}) - u_n  \right)(1+o_P(1)) = \frac{\gamma^2 S_n^{(1)}(u_n)}{(\gamma + 1)}  + o_P(1).
\end{eqnarray*}
\vspace{-0.25cm}
\newline
This shows (\ref{Hill_Hill.alt.2}). The proof of (\ref{2nd_2nd.alt.2}) is similar by noticing that $1 - z^{r/\gamma} \sim r(1-z)/\gamma$  as $z \to 1$. \qed
\end{proof}
Now we are ready to show the joint asymptotic normality of the two tail array sums involved in the likelihood moment estimation (\ref{LME1}) and (\ref{LME2}). As usual, the idea is to use the Cram\'er-Wold device.
\begin{theo} \label{Theo:Bivariate} Let the Assumptions 1-9 hold. Further assume that there exists a positive sequence $w_n$, $w_n \to \infty$ as $n \to \infty$, such that for some $0 < \varepsilon < 1/\gamma \wedge 1$:
\begin{equation}
  r_n w_n \exp \left( (\varepsilon - 1/\gamma)w_n \right) \to 0 \label{w_n4}
\end{equation}
and
\begin{equation}
\frac{r_n w_n}{\sqrt{k}} \to 0 \label{w_n5}
\end{equation}
as $n,k,n/k \to \infty$. Also, denote
\begin{eqnarray*}
&& \varsigma_1 := \sqrt{k} \left( \frac{1}{k} \sum_{j = 0}^{k-1} \log \left( 1 + \frac{\gamma Y''_j}{\sigma(n/k)}  \right)  - \gamma \right), \\[0.1cm]
&& \varsigma_ 2:= \sqrt{k} \left(  \frac{1}{k}  \sum_{j = 0}^{k-1}  \left( 1 + \frac{\gamma Y''_j}{\sigma(n/k)}  \right)^{r/\gamma} - \frac{1}{(1-r)}  \right)
\end{eqnarray*}
and $\boldsymbol{\varsigma}:= (\varsigma_1,\varsigma_2)^T$. Then
\begin{equation*}
\boldsymbol{\varsigma} \overset{D}{\to} N( \boldsymbol{0},\Sigma),
\end{equation*}
as $n,k,n/k \to \infty$, where
\begin{equation}
\Sigma := \left(
  \begin{array}{cc} \kappa_1(\gamma,r) & - \kappa_3(\gamma,r)  \\[0.25cm]- \kappa_3(\gamma,r) & \kappa_2(\gamma,r)
  \end{array}
  \right), \label{variance}
\end{equation}
with $\kappa_1(\gamma,r), \kappa_2(\gamma,r), \kappa_3(\gamma,r)$ given in (\ref{2.2.9.2})-(\ref{2.2.9.6}).
\end{theo}
\paragraph{Proof} Let $\alpha_1, \alpha_2 \in \mathbb{R} \diagdown \{ 0 \}$ and define $\boldsymbol{\alpha} :=  (\alpha_1,\alpha_2)^T$. If we are able to show that
\begin{equation}
 \boldsymbol{\alpha}^T  \boldsymbol{\varsigma} \overset{D}{\to} N(0, \boldsymbol{\alpha}^T \Sigma  \boldsymbol{\alpha}), \label{Cramer}
\end{equation}
$n,k, n/k \to \infty$, the desired result simply follows by the Cram\'er-Wold device.
\newline Recall that $Y_j = \log(|X_j|) - u_n$, $\beta'_1 = \gamma/(\gamma + 1)$, $\beta'_2 = -r/(1-r + \gamma)$ and $\tilde{\phi}_i(x) = \phi_i(x) - \beta'_{i} \mathbf{1}\{x > 0 \}$, $i=1,2$. Then, by virtue of Result \ref{Res:Variances}, Result \ref{Res:Firstdiff}, Result \ref{Res:Seconddiff}, (\ref{Expconv}) and the definition of $S_n^{(i)}(x), i = 1,2$, it is easy to see that
\vspace{-0.1cm}
\begin{eqnarray*}
\varsigma_1 &=&  \frac{1}{\sqrt{k}}   \left(   \sum_{j = 1}^{n}  \left[ \tilde{\phi}_1(Y_j) -   \E \left(  \tilde{\phi}_1(Y_1) \right) \right] \phantom{\sum_{j = 1}^{n} } \hspace{-0.55cm} \right) + o_P(1), \\[0.1cm]
\varsigma_2 &=&  - \frac{1}{\sqrt{k}}   \left(   \sum_{j = 1}^{n}  \left[ \tilde{\phi}_2(Y_j) -   \E \left(  \tilde{\phi}_2(Y_1) \right) \right] \phantom{\sum_{j = 1}^{n} } \hspace{-0.55cm} \right) + o_P(1), \\[0.1cm]
\end{eqnarray*}
Thus, writing $\tilde{\phi}(x) :=\alpha_1 \tilde{\phi}_1(x) - \alpha_2 \tilde{\phi}_2(x)$, the left hand side of (\ref{Cramer}) may be rewritten as
\begin{equation}
\boldsymbol{\alpha}^T  \boldsymbol{\varsigma} =  \frac{\tilde{\tau}_n}{\sqrt{k}}   \left( \frac{1}{\tilde{\tau}_n} \sum_{j = 1}^{n}  \left[ \tilde{\phi}(Y_j)  - \E ( \tilde{\phi}(Y_1) ) \right] \right) + o_P(1), \label{Cramer1}
\end{equation}
\vspace{-0.2cm}
\newline
where $\tilde{\tau}_n^2$ is the block-sum variance of the tail array sum with function $\tilde{\phi}$, i.e.
\vspace{-0.1cm}
\begin{equation*}
\tilde{\tau}_n^2 := \frac{n}{r_n} \Var\left(\sum_{j = 1}^{r_n} \tilde{\phi} (Y_j) \right) = \boldsymbol{\alpha}^T \Sigma_n \boldsymbol{\alpha}
\end{equation*}
with
\begin{equation*}
  \Sigma_n := \left(
  \begin{array}{cc} \left(\tilde{\tau}_n^{(1)}\right)^2 & -\tilde{\tau}_n^{(1,2)} \\[0.1cm] -\tilde{\tau}_n^{(1,2)} & \left(\tilde{\tau}_n^{(2)}\right)^2
  \end{array}
  \right).
\end{equation*}
From (\ref{2.2.9.2})-(\ref{2.2.9.6}), we simply conclude that $\tilde{\tau}_n/ \sqrt{k} = \sqrt{\boldsymbol{\alpha}^T \Sigma_n \boldsymbol{\alpha} / k}\to \sqrt{\boldsymbol{\alpha}^T \Sigma \boldsymbol{\alpha}}$ as $n,k,n/k$ $\to \infty$ and hence, by a simple application of Slutsky's Theorem on (\ref{Cramer1}), the proof is complete if we are able to show that for $n,k, n/k \to \infty$:
\vspace{-0.1cm}
\begin{equation*}
\frac{1}{\tilde{\tau}_n} \sum_{j = 1}^{n}  \left[ \tilde{\phi}(Y_j)  - \E ( \tilde{\phi}(Y_1) ) \right] \overset{D}{\to} N(0,1).
\end{equation*}
But this just means that we need to check that all the conditions up to the Lindeberg Condition of Theorem 4.1 in Rootz\'en et al. (\cite{Rootzen2}) are satisfied for the function $\psi(x) := (\tilde{\phi}(x) - \E(\tilde{\phi}(Y_1))) /\tilde{\tau}_n$. Clearly, their basic Assumptions and their (4.1) hold due to our Assumptions 5-7 and $\psi$ satisfies their (4.2) by definition. Furthermore, their negligibility conditions (2.3) hold by Lemma \ref{Lem:NegCon} below and their Corollary 2.2. so that it only suffices to check whether the Lindeberg Condition holds. Therefore, define $Z := \tilde{\tau}_n^{-1} \sum_{j=1}^{r_n}[\tilde{\phi}(Y_j) - \E(\tilde{\phi}(Y_1))]$, $Z_i := \left( \tau_n^{(i)} \right)^{-1} \sum_{j=1}^{r_n}[\phi_i(Y_j) - \E(\phi_i(Y_1))], i = 1,2$ and $Z_I := \left( \tau_n^{(I)} \right)^{-1} \sum_{j=1}^{r_n}[\mathbf{1}\left\{ Y_j > 0 \right\} $ $ - \Prob( Y_j > 0)]$. Also notice from (\ref{2.2.9.2})-(\ref{2.2.9.6}) that there always exists a $n_0 > 0$ and a constant $K > 0$ such that for any $n > n_0$: $\tau_n^{(i)}, \tau_n^{(I)} \leq K \tilde{\tau}_n$, $i =1,2$. Using this fact and with the help of their inequality (6.4) which holds for any r.v.s $X, Y$ and $\varepsilon > 0$:
\vspace{-0.1cm}
\begin{equation*}
(X + Y)^2 \mathbf{1} \{|X + Y| \geq \varepsilon \} \leq 4 (  X^2 \mathbf{1} \{|X| \geq \varepsilon/2 \} +  Y^2 \mathbf{1} \{|Y| \geq \varepsilon/2 \} ) ,
\end{equation*}
we conclude by definition of $\tilde{\phi}$ and denoting $A:= \alpha_2 \beta'_2 - \alpha_1 \beta'_1 $ that for any $n > n_0$:
\vspace{-0.15cm}
\begin{eqnarray}
&& m_n \E ( Z^2 \mathbf{1} \{|Z| > \varepsilon\}) \nonumber \\[0.25cm]
&&  =  m_n \E \left[ \left( \alpha_1 \frac{\tau_n^{(1)}}{\tilde{\tau}_n} Z_1  + \left( - \alpha_2 \frac{\tau_n^{(2)}}{\tilde{\tau}_n} Z_2 \right) + A \, \frac{\tau_n^{(I)}}{\tilde{\tau}_n} Z_I \right)^2 \right. \nonumber \\[0.1cm]
&& \hspace{+ 0.5cm}\cdot \left. \mathbf{1} \left\{\left| \alpha_1 \frac{\tau_n^{(1)}}{\tilde{\tau}_n} Z_1  + \left( - \alpha_2 \frac{\tau_n^{(2)}}{\tilde{\tau}_n} Z_2 \right) + A \,  \frac{\tau_n^{(I)}}{\tilde{\tau}_n} Z_I \right| > \varepsilon \right\}^{\phantom{2}} \hspace{-0.15cm} \right] \nonumber \\[0.1cm]
&&  \leq  4 m_n  \E \left(  \left(\alpha_1 \frac{\tau_n^{(1)}}{\tilde{\tau}_n} Z_1  + \left( - \alpha_2 \frac{\tau_n^{(2)}}{\tilde{\tau}_n} Z_2 \right) \right)^2  \mathbf{1} \left\{\left| \alpha_1 \frac{\tau_n^{(1)}}{\tilde{\tau}_n} Z_1  + \left( - \alpha_2 \frac{\tau_n^{(2)}}{\tilde{\tau}_n} Z_2 \right) \right| > \varepsilon/2 \right\} \right)   \nonumber \\[0.1cm]
&& \hspace{+ 0.5cm} + 4 m_n  \left( A \,  \frac{\tau_n^{(I)}}{\tilde{\tau}_n} \right)^2  \E \left( Z_I^2 \mathbf{1} \left\{\left| Z_I   \right| > \varepsilon \, / \left|2  A \,  \frac{\tilde{\tau}_n}{\tau_n^{(I)}} \right| \right\} \right) \nonumber \\[0.25cm]
&&  \leq  16 m_n \left(\alpha_1 \frac{\tau_n^{(1)}}{\tilde{\tau}_n} \right)^2 \E \left(  Z_1^2  \mathbf{1} \left\{\left| Z_1  \right| > \varepsilon \, / \left|4 \alpha_1 \frac{\tilde{\tau}_n}{\tau_n^{(1)}} \right|\right\} \right)   \nonumber \\[0.1cm]
&& \hspace{+ 0.5cm} + 16 m_n \left(\alpha_2 \frac{\tau_n^{(2)}}{\tilde{\tau}_n} \right)^2 \E \left(  Z_2^2  \mathbf{1} \left\{\left|  Z_2  \right| > \varepsilon \, / \left| 4 \alpha_2 \frac{\tilde{\tau}_n}{\tau_n^{(2)}} \right| \right\} \right)   \nonumber \\[0.1cm]
&& \hspace{+ 0.5cm} + 4 m_n  \left( A \,  \frac{\tau_n^{(I)}}{\tilde{\tau}_n} \right)^2  \E \left( Z_I^2 \mathbf{1} \left\{\left| Z_I   \right| > \varepsilon \, / \left|2 A \,  \frac{\tilde{\tau}_n}{\tau_n^{(I)}} \right| \right\} \right)  \nonumber \\[0.25cm]
&&  \leq \text{const.} \left[ \, m_n \E \left( Z_1^2 \mathbf{1} \{|Z_1| > \varepsilon/(4 \, |\alpha_1| K) \}\right) + m_n \E \left( Z_2^2 \mathbf{1} \{|Z_2| > \varepsilon/(4 \, |\alpha_2| K) \}\right)\right. \nonumber \\[0.25cm]
&& \hspace{+ 0.5cm} + \left. m_n \E \left( Z_I^2 \mathbf{1} \{|Z_I| > \varepsilon/(2 |A| K) \}\right) \right]. \label{Lindetwo}
\end{eqnarray}
But this simply means that we have to check whether the Lindeberg Condition holds for each tail array block-sum with respective function $(\phi_1(x) - \E(\phi_1(Y_1)))/\tau_n^{(1)}$, $(\phi_2(x) - \E(\phi_2(Y_1)))/\tau_n^{(2)}$ and $(\mathbf{1}\{x > 0\} - \Prob(Y_1 > 0))/\tau_n^{(I)} $, which is equivalent to check whether all conditions of their Theorem 6.2. are satisfied for the aforementioned functions. Due to our Assumptions 1-3, their (6.1) clearly holds for $\{\log(|X_j|)\}$ and so does their (2.4) by virtue of their Lemma 4.3. The other conditions are readily checked so that all terms in (\ref{Lindetwo}) finally converge to zero. \qed
\begin{lem} \label{Lem:NegCon}  Suppose the Assumptions 1-7 hold. Then, using the very same definitions and notations as in the proof of Theorem \ref{Theo:Bivariate},
\begin{align}
& \frac{n}{r_n \tilde{\tau}^2_n} \Var\left(\sum_{j = 1}^{l_n} \tilde{\phi} (Y_j) \right) \to 0, \label{2.2.15.1} \\[0.1cm]
&\tilde{\tau}_{n}^{-2} \Var\left(\sum_{j = 1}^{n -r_n m_n} \tilde{\phi} (Y_j) \right) \to 0,  \label{2.2.15.2}
\end{align}
as $n,k,n/k \to \infty$.
\end{lem}
\begin{proof}
To prove (\ref{2.2.15.1}), split up the variance into parts of variances and covariances:
\begin{eqnarray*}
&& \frac{n}{r_n \tilde{\tau}^2_{n}} \Var\left(\sum_{j = 1}^{l_n} \tilde{\phi} (Y_j) \right) \\[0.1cm] &=&
\frac{n}{r_n \tilde{\tau}^2_{n}} \left[ \alpha_1^2 \Var\left(\sum_{j = 1}^{l_n} \phi_1(Y_j) \right)  + \alpha_2^2 \Var\left(\sum_{j = 1}^{l_n} \phi_2(Y_j) \right) + A^2 \Var\left(\sum_{j = 1}^{l_n}  \mathbf{1}\{Y_j > 0 \} \right) \right. \\[0.1cm] && -2\alpha_1 \alpha_2 \Cov\left(\sum_{j = 1}^{l_n} \phi_1(Y_j) , \sum_{j = 1}^{l_n} \phi_2(Y_j) \right)  + 2 A \alpha_1 \Cov\left(\sum_{j = 1}^{l_n} \phi_1(Y_j), \sum_{j = 1}^{l_n}  \mathbf{1}\{Y_j > 0 \} \right)
\\[0.1cm] && \left. \hspace{-1mm}- 2A \alpha_2 \Cov\left(\sum_{j = 1}^{l_n} \phi_2(Y_j), \sum_{j = 1}^{l_n}  \mathbf{1}\{Y_j > 0 \} \right) \right].
\end{eqnarray*}
Now, recall that all conditions of Lemma 4.3 in \cite{Rootzen2} are satisfied for each of the three functions $(\phi_1(x) - \E(\phi_1(Y_1)))/\tau_n^{(1)}$, $(\phi_2(x) - \E(\phi_2(Y_1)))/\tau_n^{(2)}$ and $(\mathbf{1}\{x > 0\} - \Prob(Y_1 > 0))/\tau_n^{(I)} $ and since by definition $\tilde{\tau}_n$ is a linear combination of $\tau_n^{(1)}, \tau_n^{(2)}, \tau_n^{(I)}, \tau_n^{(1,2)}, \tau_n^{(I,1)}$ and $\tau_n^{(I,2)}$ and all elements are asymptotically of the same order with different constants by Result \ref{Res:Variances}, all variance terms above clearly converge to zero as $n,k,n/k \to \infty$. Then, for the first covariance term:
\vspace{-0.2cm}
\begin{eqnarray*}
&& \frac{n}{r_n \tilde{\tau}^2_{n}}  \Cov\left(\sum_{j = 1}^{l_n} \phi_1(Y_j) , \sum_{j = 1}^{l_n} \phi_2(Y_j) \right) \\ &=&  \frac{n}{r_n \tilde{\tau}^2_{n}}  \sum_{j = 1}^{l_n} \Cov\left(\phi_1(Y_j), \phi_2(Y_j) \right)   + \frac{n}{r_n\tilde{\tau}^2_{n}}   \sum_{i < j}^{l_n} \Cov\left(\phi_1(Y_i), \phi_2(Y_j) \right)  \\ &&  + \frac{n}{r_n\tilde{\tau}^2_{n}}   \sum_{i > j}^{l_n} \Cov\left(\phi_1(Y_i), \phi_2(Y_j) \right)    :=  \text{I}_1 + \text{I}_2 + \text{I}_3.
\end{eqnarray*}
Going along the steps of \cite{Resnick4}, p. 708-711, it is easy to see that
\begin{eqnarray*}
\text{I}_1 &=&  \frac{l_n}{r_n} \frac{k}{\tilde{\tau}^2_{n}} \frac{n}{k} \Cov\left( \phi_1(Y_1), \phi_2(Y_1)  \right) \\[0.1cm] &=& o(1)O(1)O(1) = o(1)
\end{eqnarray*}
as $n,k,n/k \to \infty$ because $l_n/r_n = o(1)$, $k/\tilde{\tau}^2_{n} = O(1)$ and $ \frac{n}{k}\Cov\left( \phi_1(Y_1), \phi_2(Y_1)  \right) \sim  \frac{n}{k}\E\left( \phi_1(Y_1)\cdot\phi_2(Y_1)  \right) = \int_0^{\infty} \frac{n}{k} \Prob(\phi_1(Y_1)\phi_2(Y_2) > x) dx =  O(1)$.  Also, by the stationarity of  $\{Y_j\}_{j \geq 1}$:
\vspace{-0.2cm}
\begin{eqnarray*}
\text{I}_2 &=& \frac{n}{r_n \tilde{\tau}^2_{n}} \sum_{j=1}^{l_n - 1 } (l_n - j) \Cov \left(\phi_1(Y_1), \phi_2(Y_{j+1}) \right) \\[0.1cm]
&=&  \frac{n\, l_n }{r_n \tilde{\tau}^2_{n}} \sum_{j=1}^{l_n - 1 }  \E \left(\phi_1(Y_1), \phi_2(Y_{j+1}) \right) + o(1)  \\[0.1cm]
&=& \frac{k}{\tilde{\tau}^2_{n}} \, \frac{l_n}{r_n}\frac{n}{k} \sum_{j=1}^{l_n - 1 } \E \left(\phi_1(Y_1), \phi_2(Y_{j+1}) \right)+ o(1) \\[0.25cm]
&=& O(1)o(1)O(1) + o(1) = o(1),
\end{eqnarray*}
$n,k,n/k \to \infty$. Similarly, one shows $\text{I}_3 \to 0$ as $n,k,n/k \to \infty$ so that the whole term asymptotically  vanishes. Analogous considerations for the other covariance terms conclude the proof of (\ref{2.2.15.1}). To show (\ref{2.2.15.2}), notice that $0 < n - r_n m_n < r_n$, which simply yields $n - r_n m_n = o(n).$ Thus, the proof is straightforward by simply repeating the previous steps. \qed
\end{proof}
Now we can state our main theorem:
\begin{theo} \label{Theo:Asnorm} Let the Assumptions 1-9 hold and $r < 0$. Further suppose there exists a sequence $w_n$ such that (\ref{w_n4}) and (\ref{w_n5}) hold. Then
\begin{equation*}
\sqrt{k} \left( \begin{array}{c} \displaystyle \hat{\gamma}_{LME} - \gamma \\[0.25cm] \displaystyle \frac{\hat{\sigma}_{LME}}{\sigma(n/k)} - 1 \end{array} \right) \overset{D}{\to} N\left( \boldsymbol{0},L \Sigma L^T\right)
\end{equation*}
as $n,k, n/k \to \infty$, where
\begin{equation*} L :=
\left(
\begin{array}{cc} \displaystyle
-\frac{(1-r)(1+\gamma)}{\gamma r} & \displaystyle   \frac{(1-r)^2(1+\gamma -r)}{r^2}  \\[0.5cm]
\displaystyle \frac{(1+\gamma)}{\gamma r}
& \displaystyle - \frac{(1-r)^2(1+\gamma-r)}{r^2}
\end{array}
\right),
\end{equation*}
and $\Sigma$ as in (\ref{variance}).
\end{theo}
\begin{proof} Recall that $Y''_i = |X|_{n,n-i}- |X|_{n,n-k}$ and denote $\mathbf{Y''} := \{ Y''_0,Y''_1,\ldots,Y''_{k-1} \}$. Next define for $x,y > 0$:
\begin{equation*}
  \mathbf{Z}(\mathbf{Y''},x,y):= \left( \begin{array}{c} \frac{1}{k} \sum_{i=0}^{k-1} \log \left( 1 + \frac{x Y''_i}{y \sigma(n/k)} \right) - x \\[0.1cm]  \frac{1}{k} \sum_{i=0}^{k-1}  \left( 1 + \frac{x Y''_i}{y\sigma(n/k)} \right)^{r/x} - 1/(1-r)
  \end{array} \right)
\end{equation*}
Notice that $ \mathbf{Z}(\mathbf{Y''},\hat{\gamma}_{LME},\hat{\sigma}_{LME}/\sigma(n/k)) = (0,0)^T$ due to (\ref{LME1}) and (\ref{LME2}) and thus, by the mean value theorem of differentiation, we obtain the following system of equations:
\begin{eqnarray}
&& \mathbf{Z}(\mathbf{Y''},\gamma,1) -  \mathbf{Z}(\mathbf{Y''},\hat{\gamma}_{LME},\hat{\sigma}_{LME}/\sigma(n/k)) \nonumber \\[0.25cm]&=& \mathbf{Z}(\mathbf{Y''},\gamma,1) = -\operatorname{D}\hspace{-0.1cm}\mathbf{Z}(\mathbf{Y''},\gamma_n,\sigma_n)\cdot (\Delta\gamma,\Delta\sigma)^T, \label{MVT3}
\end{eqnarray}
where $(\gamma_n,\sigma_n)^T = \alpha_n (\gamma,1)^T + (1-\alpha_n) (\hat{\gamma}_{LME},\hat{\sigma}_{LME}/\sigma(n/k))^T$ for all $n$ with some $\alpha_n \in (0,1)$, $\Delta\gamma := \hat{\gamma}_{LME} - \gamma$, $\Delta \sigma := \hat{\sigma}_{LME}/ \sigma(n/k) -1$ and
\begin{eqnarray*}
&&\operatorname{D}\hspace{-0.1cm}\mathbf{Z}(\mathbf{Y''},x,y) := \nonumber \\[0.25cm]
&& \left(
\begin{array}{cc}
\frac{1}{k} \sum_{i=0}^{k-1} \frac{Y''_i}{y \sigma(n/k)} \left( 1 + \frac{x Y''_i}{y \sigma(n/k)} \right)^{-1} - 1 & -\frac{1}{k} \sum_{i=0}^{k-1}  \frac{x Y''_i}{y^2 \sigma(n/k)}  \left( 1 + \frac{x Y''_i}{y \sigma(n/k)} \right)^{-1} \\[0.5cm]
\begin{array}{l} \frac{1}{k} \sum_{i=0}^{k-1} \left[ \frac{r Y''_i}{x y \sigma(n/k)} \left( 1 + \frac{x Y''_i}{y \sigma(n/k)} \right)^{r/x-1} \right. \\ \left. -\frac{r}{x^2} \left( 1 + \frac{x Y''_i}{y \sigma(n/k)} \right)^{r/x} \log\left(  1 + \frac{x Y''_i}{y \sigma(n/k)}  \right) \right]
\end{array}
& -\frac{1}{k} \sum_{i=0}^{k-1}  \frac{r Y''_i}{y^2 \sigma(n/k)} \left( 1 + \frac{x Y''_i}{y \sigma(n/k)} \right)^{r/x - 1}
\end{array}
\right). \nonumber \\[0.1cm]  &&
\end{eqnarray*}
\vspace{-0.75cm}
\newline
Using Theorem 1 in \cite{Martig} and repeating the very same steps as in the proof of their Lemma 3, several applications of the continuous mapping theorem yield
\begin{equation}
 \operatorname{D} \hspace{-0.1cm}\mathbf{Z}(\mathbf{Y''},\gamma_n,\sigma_n) \overset{P}{\to} \left(
\begin{array}{cc} \displaystyle
-\frac{\gamma}{(1 + \gamma)} & \displaystyle  -\frac{\gamma}{(1 + \gamma)}  \\[0.5cm]
\displaystyle  \frac{-r}{(1-r)^2 (1 + \gamma -r)}
& \displaystyle  \frac{-r}{(1-r)(1+\gamma -r)}
\end{array}
\right)  \label{M} \end{equation}
as $n,k, n/k \to \infty$.
\newline Notice that the matrix in (\ref{M}) is non-singular (remember that $r < 0$) and hence, by \cite{Lehmann}, Lemma 5.3.3, it follows that
\begin{equation}
-\operatorname{D}\hspace{-0.1cm}\mathbf{Z}(\mathbf{Y''},\gamma_n,\sigma_n)^{-1} \overset{P}{\to} \left(
\begin{array}{cc} \displaystyle
-\frac{(1-r)(1+\gamma)}{\gamma r} & \displaystyle   \frac{(1-r)^2(1+\gamma -r)}{r^2}  \\[0.5cm]
\displaystyle \frac{(1+\gamma)}{\gamma r}
& \displaystyle - \frac{(1-r)^2(1+\gamma-r)}{r^2}
\end{array}
\right) \label{M_inv}
\end{equation}
as $n,k, n/k \to \infty$.
\newline Now, a simple reformulation of (\ref{MVT3}) yields
\begin{equation*}
 \sqrt{k} (\Delta\gamma,\Delta\sigma)^T = \sqrt{k} \mathbf{Z}(\mathbf{Y''},\gamma,1) \left(- \operatorname{D}\hspace{-0.1cm}\mathbf{Z}(\mathbf{Y''},\gamma_n,\sigma_n)^{-1} \right),
\end{equation*}
and thus by (\ref{M_inv}) and Theorem \ref{Theo:Bivariate}:
\begin{equation*}
\sqrt{k} \mathbf{Z}(\mathbf{Y''},\gamma,1) \overset{D}{\to} N( \boldsymbol{0} ,\Sigma),
\end{equation*}
 as $n,k, n/k \to \infty$. The desired result follows by Slutsky's Theorem. \qed
\end{proof}
In the last part of this section, we will  simplify and generalize the conditions listed in our main theorem as far as possible. First of all notice that if our Assumptions 1-3 hold, then Assumption 8 can be reformulated as follows:
\vspace{0.25cm}
\newline Assume $\F_{X}$ has a derivative $\F'_{X}$ on $\mathbb{R} \diagdown [-t_0, t_0]$ where $t_0 > 0$. Then (A.8) holds if
\begin{equation}
\lim_{t \to \infty}\frac{t \F'_{X}(t)}{1-\F_{X}(t)} = \frac{1}{\gamma}. \label{vMises2}
\end{equation}
Next consider the following assumptions for the cdf of the innovations:
\vspace{0.25cm}
\newline Assume there exists $d < 1$ such that
\vspace{-0.1cm}
\begin{equation}
\E ( Z_1^d) < \infty \label{d-Moment}
\end{equation}
\vspace{-0.3cm}
\newline
and that $\G_Z$ has a derivative $\G'_Z$ which is $L_1$-Lipschitz, that is for any $y > 0$ and some positive constant $C$:
\begin{equation}
\int_0^{\infty} |\G'_Z(x) - \G'_Z(x+y)| \, dx < Cy. \label{Lipschitz}
\end{equation}
Finally assume that
\begin{equation}
\liminf_{n,k,n/k \to \infty} n/k^{3/2} > 0 \hspace{0.5cm} \text{or} \hspace{0.5cm}  \limsup_{n,k,n/k \to \infty} n/k^{3/2} < \infty. \label{limsup}
\end{equation}
From \cite{Resnick4}, pp. 712-713, we derive that if our Assumptions 1 and 2 as well as (\ref{vMises2})-(\ref{limsup}) hold, then our Assumptions 5-7 as well as (\ref{w_n4}) and (\ref{w_n5}) are always satisfied. Let's summarize this in the following corollary:
\begin{cor} \label{Cor:Asnorm}
Let the Assumptions 1-4 and 9 as well as (\ref{vMises2})-(\ref{limsup}) hold. Then for $r < 0$:
\begin{equation*}
\sqrt{k} \left( \begin{array}{c} \displaystyle \hat{\gamma}_{LME} - \gamma \\[0.25cm] \displaystyle \frac{\hat{\sigma}_{LME}}{\sigma(n/k)} - 1 \end{array} \right) \overset{D}{\to} N\left( \boldsymbol{0},L \Sigma L^T\right)
\end{equation*}
as $n,k, n/k \to \infty$.
\end{cor}
\section{Checking the conditions of the main theorem}
\setcounter{equation}{0}
Even tough Corollary \ref{Cor:Asnorm} helps for a better understanding under what circumstances bivariate asymptotic normality holds, it still has some assumptions which are rather uneasy to check. Both the von Mises-Condition (\ref{vMises2}) and the second-order (extended) regular variation conditions give constraints on the tail distribution function (tdf) of the linear process $1 - \F_{|X|}$. Since this tdf depends on $1-\G_Z$ and $\{c_j\}_{j \geq 0}$ (see (\ref{Result 2.0.1})), it is of practical and statistical advantage to reformulate these unpleasant assumptions using only $1-\G_Z$ (possibly adding some other mild conditions). In other words, we want to formulate Corollary \ref{Cor:Asnorm} under assumptions on the tdf of the innovations.
\newline A very powerful tool to remedy this problem is to do asymptotic expansions as proposed in \cite{Barbe}. In their Proposition 4.2.2 (which is an application of their main Theorem 2.5.1) they show that if $1 - \G_Z$  satisfies some smoothness conditions and is second-order regularly varying, then $1-\F_{|X|}$ obeys the von Mises-Condition and is second-order regularly varying itself! This very nice result is indeed not far away from our target (we also need second-order extended regular variation beside the ``normal'' second-order regular variation) but in fact, as we will see in Example \ref{Ex:Pareto}, it is possible to expand the results of \cite{Barbe} in such a manner that all conditions of Corollary \ref{Cor:Asnorm} are satisfied. Nevertheless it would be of great interest to know whether it is possible to expand the above listed results even to second-order extended regular variation to state a complete and practically implementable theorem for bivariate asymptotic normality of the likelihood moment estimators.
\vspace{0.5cm}
\newline In the following example we  show that the assumptions of Corollary \ref{Cor:Asnorm} hold if the cdf of the innovations is Pareto. Let for convenience $\alpha := 1/\gamma$:
\begin{ex} \rm \label{Ex:Pareto} Assume that $\G_Z(x) = 1-x^{-\alpha}, \, x \geq 1$ and $\alpha > 2$. Let $X_n = \sum_{j=0}^{\infty} c_j$ $ Z_{n-j}$ be a weakly stationary linear process consisting of a sequence of \emph{non-negative} coefficients $\{c_j\}_{j \geq 0}$, at least one $c_j >0$, satisfying our Assumptions 3 and 4. This will also ensure that (\ref{d-Moment}) and (\ref{Lipschitz}) trivially hold and that $|X_n| = X_n $. Finally recall that $\E(Z_1) = \alpha/(\alpha - 1) := \mu > 0$ and that $\Var(Z_1) =  \alpha/((\alpha-1)(\alpha-2)) := \sigma^2$.
\vspace{0.25cm}
\newline Next we are going to apply Theorem 2.5.1 of \cite{Barbe} twice in order to obtain asymptotic expansions for $1 - \F_{X}$ and $(1 - \F_{X})'$:
\newline Using their notations, take $k = 0$, $m = 2$ and $\omega = 3$ in the first case and $k = 1$, $m = 0$ and $\omega =2$ in the second case and notice that $\G_Z$ is smoothly varying of index $-\alpha$ and order $\omega = 3$. Further denote $C_u := \sum_{j = 0}^{\infty} c_j^u$ and assume there exists $0 < \xi < 1$ such that for $\eta :=  \xi (\alpha/(\alpha + 3) \wedge 1/2)$:
\begin{equation*}
C_{\eta} < \infty. \tag{i}  \\[0.25cm ]
\end{equation*}
Also, let
\begin{equation*}
(C_2 C_{\alpha} - C_{\alpha + 2})\sigma^2 + (C_1^2 C_{\alpha} - 2 C_1 C_{\alpha + 1} + C_{\alpha +2}) \mu^2 \neq 0. \tag{ii}
\end{equation*}
Then for $t\to\infty$
\begin{eqnarray}
1 - \F_{X}(t) &=& C_{\alpha} t^{-\alpha} + \alpha \mu (C_1 C_{\alpha} - C_{\alpha + 1}) t^{-\alpha - 1} +  \frac{1}{2} \alpha (\alpha + 1)\cdot \nonumber \\[0.1cm]  && \hspace{+0.1cm} \left[ (C_2 C_{\alpha} - C_{\alpha + 2})\sigma^2  +  (C_1^2 C_{\alpha} - 2 C_1 C_{\alpha + 1} + C_{\alpha +2}) \mu^2 \right]  t^{-\alpha - 2} \nonumber \\[0.15cm] && \hspace{+0.1cm} +   o(t^{-\alpha - 2})  \nonumber \\[0.25cm]
&:=& \tilde{c}_1 t^{-\alpha} + \tilde{c}_2 t^{-\alpha - 1} + \tilde{c}_3 t^{-\alpha - 2}  +  o(t^{-\alpha - 2}),  \label{TailApprox}
\end{eqnarray}
where $\tilde{c}_1, \tilde{c}_3 \neq 0$, and
\begin{equation}
(1 - \F_{X}(t))' = - \alpha C_{\alpha} t^{-\alpha - 1} + o(t^{-\alpha - 1}). \label{vonMisesApprox}
\end{equation}
 Now, from (\ref{TailApprox}) and (\ref{vonMisesApprox}), it simply follows that
\begin{equation*}
\lim_{t \to \infty} \frac{t \F'_X(t)}{1-\F_X(t)} = \lim_{t \to \infty} \frac{t  \alpha C_{\alpha} t^{-\alpha - 1} (1+o(1))}{C_{\alpha} t^{-\alpha}(1+o(1))} = \alpha,
\end{equation*}
so that the von Mises-Condition (\ref{vMises2}) obviously holds.
\vspace{0.25cm}
\newline In order to check whether the second-order extended regular and second-order regular variation assumptions hold, we need an asymptotic expansion for $ b(x) := (1/(1-\F_{X}))^{\leftarrow}(x)$. After some calculations, we deduce that
\begin{equation}
b(x) = a_1 x^{1/\alpha}  + a_2  +  a_3 x^{-1/\alpha} + o(x^{-1/\alpha}), \label{three-term}
\end{equation}
where
\begin{equation*}
a_1 = \tilde{c}_1^{1/\alpha}, \hspace{0.5cm} a_2 = \tilde{c}_2/(\alpha \tilde{c}_1), \hspace{0.5cm} a_3 =  - \tilde{c}_1^{-1/\alpha - 2} \left[ (1+\alpha)\tilde{c}_2^2  /(2 \alpha) - \tilde{c}_1 \tilde{c}_3 \right] / \alpha
\end{equation*}
and because we want $a_3$ to be non-zero we also assume that
\begin{equation*}
(1 + \alpha) \tilde{c}_2^2/(2\alpha) - \tilde{c}_1\tilde{c}_3 \neq 0. \tag{iii}
\end{equation*}
Using (\ref{three-term}), it is now easy to see that $b(t)$ is 2ERV$(-2/\alpha, 2 a_1^{-1} a_3 t^{-2/\alpha}/\alpha)$ and our $k$ has to be chosen such that for $n,k,n/k \to \infty$:
\begin{equation}
\sqrt{k} A(n/k) = \text{const.} \sqrt{k} \, (n/k)^{-2/\alpha} \to 0.
\label{convspeed1}
\end{equation}
To show that $1 - \F_{X}$ is also second-order regularly varying, we have to deal with two cases:
\newline If $\tilde{c}_2 = 0$, then, according to section 2, $1 - \F_{X}$ is 2RV$(-2, -2\tilde{c}_3t^{-2}/\tilde{c}_1)$. Hence the rate of convergence is the same as for second-order extended regular variation and the only constraints on our sequence $k$ are (\ref{limsup}) and (\ref{convspeed1}).
\newline In case $\tilde{c}_2 \neq 0$ the rates of convergence required by the two forms of second-order condition are not the same because in this particular case $1 - \F_{X}$ is 2RV$(-1, -\tilde{c}_2t^{-1}/\tilde{c}_1)$ and $k$ has to be chosen such that for $n,k,n/k \to \infty$:
\begin{equation}
\sqrt{k} A^*(b_{X}(n/k)) = \text{const.} \sqrt{k} \, b^{-1}_{X}(n/k)  \sim \text{const.} \sqrt{k} \, (n/k)^{-1/\alpha}  \to 0.
\label{convspeed3}
\end{equation}
Since every choice of $k$ which fulfils (\ref{convspeed3}) automatically satisfies (\ref{convspeed1}), the only constraints on our sequence $k$ in case that $\tilde{c}_2 \neq 0$ are those in (\ref{limsup}) and (\ref{convspeed3}).
\newline Thus, a suitable choice of $k$ is
\begin{equation*}
k = k(n) := \begin{cases} n^{2\theta/(2 + \alpha)} & \text{if} \: \tilde{c}_2 \neq 0 \\
n^{4\theta/(4 + \alpha)} &\text{if} \: \tilde{c}_2 = 0, \end{cases} \tag{iv}
\end{equation*}
where $\theta \in (0,1)$ is an arbitrary constant.
\vspace{0.25cm}
\newline Hence, assuming the mild conditions (i)-(iv) to hold, all assumptions of Corollary \ref{Cor:Asnorm} are satisfied for this very special setup. \hfill  	$\bigtriangleup$ \end{ex}
\section{Acknowledgments} The authors want to thank William P. McCormick for his contributions and very thorough explanations regarding the asymptotic tail expansions.

\newpage

\end{document}